\documentclass[11pt,reqno]{amsart}

\usepackage{amscd,amssymb,amsmath,amsthm}
\usepackage[arrow,matrix]{xy}
\usepackage{graphicx}
\usepackage{color}
\usepackage{cite}
\newtheorem{thm}{Theorem}

\newtheorem{lemma}{Lemma}
\newtheorem{pro}{Proposition}
\newtheorem{rk}{Remark}
\newtheorem{cor}{Corollary}

\numberwithin{equation}{section} 

\begin{document}
\title[Translation-invariant $p$-adic Gibbs measures for the Potts model]{
Description of all translation-invariant $p$-adic Gibbs measures
for the Potts model on a Cayley tree}

\author{U. A. Rozikov, O. N. Khakimov}

\address{U.\ A.\ Rozikov\\ Institute of mathematics,
29, Do'rmon Yo'li str., 100125, Tashkent, Uzbekistan.}
\email {rozikovu@yandex.ru}

\address{O.N. Khakimov\\ Institute of mathematics,
29, Do'rmon Yo'li str., 100125, Tashkent, Uzbekistan.}
\email {hakimovo@mail.ru}

\begin{abstract} Recently it was proved that usual (real) Potts model on a Cayley tree has up
to $2^q-1$ translation-invariant Gibbs measures. This paper is devoted to description of translation-
invariant $p$-adic Gibbs measures (TIpGMs) of the $p$-adic Potts model. In particular, for the Cayley tree
of order two we give exact number of such measures. Mereover we give criterion of boundedness of TIpGMs

\end{abstract}
\maketitle


\section{introduction}

The $p$-adic numbers were first introduced by the German
mathematician K.Hensel. For about a century after the discovery of
$p$-adic numbers, they were mainly considered objects of pure
mathematics. However, numerous applications of these numbers to
theoretical physics have been proposed papers \cite{14},\cite{49}
to quantum mechanics and to $p$-adic valued physical
observables \cite{k91}. A number of $p$-adic models in physics
cannot be described using ordinary probability theory based on the
Kolmogorov axioms.

In \cite{33} a theory of stochastic processes with values in
$p$-adic and more general non-Archimedean fields was developed,
having probability distributions with non-Archimedean values.

One of the basic branches of mathematics lying at the base of the
theory of statistical mechanics is the theory of probability and
stochastic processes. Since the theories of probability and
stochastic processes in a non-Archimedean setting have been
introduced, it is natural to study problems of statistical
mechanics in the context of the $p$-adic theory of probability.

We note that $p$-adic Gibbs measures were studied for several $p$-adic models of statistical mechanics \cite{GRR,16,TMP2},\cite{{MRasos,37,MNGM,Fg,RK,bookroz}}.
It is known that \cite{MRasos} there exist phase transition for the $q$-state $p$-adic Potts model on the Cayley tree of order $k$ if and only if $q\in p\mathbb N$. In this paper, we shall fully describe the set of TIpGMs for the $q$-state Potts model on a Cayley tree of order two.

Our main result of this paper
is the characterization and counting of TIpGMs which is given in Theorems \ref{tigm1} and \ref{tigm2}. Let us outline the proof.
Our analysis is based on a systematic investigation of the tree recursion for boundary fields
(boundary laws) whose fixed points are
characterizing the TIpGMs. In this analysis we find all fixed points.
We show that these fixed points can be characterized according to the number of their non-zero
components, see Theorem \ref{Theorem2}. Care is needed, since not all of these solutions give rise to different Gibbs measures, and we have to take into account of symmetries in a proper way
when going from the full description of fixed points to the full description of TIpGMs. \bigskip

\section{definitions and preliminary results}

\subsection{\bf $p$-adic numbers and measures.}

Let $\mathbb Q$ be the field of rational numbers.
For a fixed prime number $p$, every rational number $x\ne0$ can
be represented in the form $x = p^r{n\over m}$, where $r,n\in\mathbb Z$, $m$ is a positive integer, and $n$ and $m$ are relatively
prime with $p$: $(p, n) = 1$, $(p, m) = 1$. The $p$-adic norm of
$x$ is given by
$$
|x|_p= \left\{\begin{array}{ll}
p^{-r},& \text{ if }  x\neq 0,\\
0,& \text{ if } x=0.
 \end{array}\right.
$$

This norm is non-Archimedean  and satisfies the so called strong
triangle inequality
$$|x+y|_p\leq \max\{|x|_p,|y|_p\}.$$

From this property immediately follow the following facts:

1) if  $|x|_p\neq |y|_p$, then $|x-y|_p=\max\{|x|_p,|y|_p\}$;

2) if  $|x|_p=|y|_p$, then  $|x-y|_p\leq |x|_p$;

The completion of $\mathbb Q$ with respect to the $p$-adic norm defines
the $p$-adic field $\mathbb Q_p$ (see \cite{29}).

The completion of the field of rational numbers $\mathbb Q$ is either
the field of real numbers $\mathbb R$ or one of the fields of
$p$-adic numbers $\mathbb Q_p$ (Ostrowski's theorem).

Any $p$-adic number $x\ne 0$ can be uniquely represented
in the canonical form
\begin{equation}\label{ek}
x = p^{\gamma(x)}(x_0+x_1p+x_2p^2+\dots),
\end{equation}
where $\gamma=\gamma(x)\in \mathbb Z$ and the integers $x_j$ satisfy: $x_0 > 0$,
$0\leq x_j \leq p - 1$ (see \cite{29,sc,48}). In this case $|x|_p =
p^{-\gamma(x)}$.

\begin{thm}\label{tx2}\cite{48}
The equation
$x^2 = a$, $0\ne a =p^{\gamma(a)}(a_0 + a_1p + ...), 0\leq a_j
\leq p - 1$, $a_0 > 0$ has a solution $x\in \mathbb Q_p$ iff hold true the following:

i) $\gamma(a)$ is even;

ii) $y^2=a_0(\operatorname{mod} p)$ is solvable for $p\ne 2$;
the equality $a_1=a_2=0$ holds if $p=2$.
\end{thm}
For $a\in \mathbb Q_p$ and $r> 0$ we denote
$$B(a, r) = \{x\in \mathbb Q_p : |x-a|_p < r\}.$$

$p$-adic {\it logarithm} is defined by the series
$$\log_p(x) =\log_p(1 + (x-1)) =
\sum_{n=1}^\infty (-1)^{n+1}{(x-1)^n\over n},$$ which converges
for  $x\in B(1, 1)$;  $p$-adic {\it exponential} is defined by
$$\exp_p(x) =\sum^\infty_{n=0}{x^n\over n!},$$
which converges
for $x \in B(0, p^{-1/(p-1)})$.
\begin{lemma}\label{el} Let $x\in B(0, p^{-1/(p-1)})$. Then
$$|\exp_p(x)|_p = 1,\ \ |\exp_p(x)-1|_p = |x|_p, \ \ |\log_p(1 + x)|_p = |x|_p,$$
$$\log_p(\exp_p(x)) = x,\ \ \exp_p(\log_p(1 + x)) = 1 + x.$$
\end{lemma}

A more detailed description of $p$-adic calculus and $p$-adic mathematical physics can be found in \cite{29,sc,48}.

Let $(X,{\mathcal B})$ be a measurable space, where ${\mathcal B}$
is an algebra of subsets of $X$. A function $\mu: {\mathcal B}\to
\mathbb Q_p$ is said to be a $p$-adic measure if for any $A_1, . . .
,A_n\in {\mathcal B}$ such that $A_i\cap A_j = \emptyset$, $i\ne
j$, the following holds:
$$\mu(\bigcup^n_{j=1}A_j)=\sum^n_{j=1}\mu(A_j).$$
A $p$-adic measure is called a probability measure if $\mu(X) =
1$. A $p$-adic probability measure $\mu$ is called {\it bounded}
if $\sup\{|\mu(A)|_p:A\in\mathcal B\}<\infty$ (see, \cite{k91}).

We call a $p$-adic measure a probability measure \cite{16} if $\mu(X)=1$.

\subsection{\bf Cayley tree.}

The Cayley tree $\Gamma^k$
of order $ k\geq 1 $ is an infinite tree, i.e., a graph without
cycles, such that exactly $k+1$ edges originate from each vertex.
Let $\Gamma^k=(V, L)$ where $V$ is the set of vertices and  $L$ the set of edges.
Two vertices $x$ and $y$ are called {\it nearest neighbors} if there exists an
edge $l \in L$ connecting them.
We shall use the notation $l=\langle x,y\rangle$.
A collection of nearest neighbor pairs $\langle x,x_1\rangle, \langle x_1,x_2\rangle,...,\langle x_{d-1},y\rangle$ is called a {\it
path} from $x$ to $y$. The distance $d(x,y)$ on the Cayley tree is the number of edges of the shortest path from $x$ to $y$.

For a fixed $x^0\in V$, called the root, we set
\begin{equation*}
W_n=\{x\in V\,| \, d(x,x^0)=n\}, \qquad V_n=\bigcup_{m=0}^n W_m
\end{equation*}
and denote
$$
S(x)=\{y\in W_{n+1} :  d(x,y)=1 \}, \ \ x\in W_n, $$ the set  of {\it direct successors} of $x$.

Let $G_k$ be a free product of $k + 1$ cyclic groups of the second order with generators $a_1, a_2,\dots, a_{k+1}$,
respectively.
It is known that there exists a one-to-one correspondence between the set of vertices $V$ of the Cayley tree $\Gamma^k$ and the group $G_k$.

\subsection{\bf $p$-adic Potts model}

Let $\mathbb Q_p$ be the field of $p$-adic numbers and
$\Phi$ be a finite set. A configuration $\sigma$ on $V$ is then defined
as a function $x\in V\to\sigma(x)\in\Phi$; in a similar fashion
one defines a configuration $\sigma_n$ and $\sigma^{(n)}$ on $V_n$
and $W_n$ respectively. The set of all configurations on $V$
(resp. $V_n,\ W_n$) coincides with $\Omega=\Phi^V$
(resp.$\Omega_{V_n}=\Phi^{V_n},\ \Omega_{W_n}=\Phi^{W_n}$). Using
this, for given configurations $\sigma_{n-1}\in\Omega_{V_{n-1}}$
and $\sigma^{(n)}\in\Omega_{W_n}$ we define their concatenations
by
$$
(\sigma_{n-1}\vee\sigma^{(n)})(x)=\left\{\begin{array}{ll}
\sigma_{n-1}(x),& \text{if}\  x\in V_{n-1},\\
\sigma^{(n)}(x),& \text{if}\  x\in W_n.
 \end{array}\right.
$$
It is clear that $\sigma_{n-1}\vee\sigma^{(n)}\in\Omega_{V_n}.$\\

Let $G^*_k$ be a subgroup of the group $G_k$. A function $h_x$ (for example, a configuration $\sigma(x)$) of $x\in G_k$ is called $G^*_k$-periodic if $h_{yx}=h_x$ (resp. $\sigma(yx)=\sigma(x)$) for any $x\in G_k$ and
$y\in G^*_k.$

A $G_k$-periodic function is called {\it translation-invariant}.

We consider {\it $p$-adic Potts model} on a Cayley tree,
where the spin takes values in the set
$\Phi:=\{1,2,\dots,q\}$, and is assigned to the vertices
of the tree.

The (formal) Hamiltonian of $p$-adic Potts model is
\begin{equation}\label{ph}
H(\sigma)=J\sum_{\langle x,y\rangle\in L}
\delta_{\sigma(x)\sigma(y)},
\end{equation}
where $J\in B(0, p^{-1/(p-1)})$ is a coupling constant,
$\langle x,y\rangle$ stands for nearest neighbor vertices and $\delta_{ij}$ is the Kroneker's
symbol:
$$\delta_{ij}=\left\{\begin{array}{ll}
0, \ \ \mbox{if} \ \ i\ne j\\[2mm]
1, \ \ \mbox{if} \ \ i= j.
\end{array}\right.
$$

\subsection{\bf $p$-adic Gibbs measure}
Define a finite-dimensional distribution of a $p$-adic probability measure $\mu$ in the volume $V_n$ as
\begin{equation}\label{p*}
\mu_{\tilde h}^{(n)}(\sigma_n)=Z_{n,\tilde h}^{-1}\exp_p\left\{H_n(\sigma_n)+\sum_{x\in W_n}{\tilde h}_{\sigma(x),x}\right\},
\end{equation}
where  $Z_{n,\tilde h}$ is the normalizing factor, $\{{\tilde h}_x=({\tilde h}_{1,x},\dots, {\tilde h}_{q,x})\in\mathbb Q_p^q, x\in V\}$ is a collection of vectors and
$H_n(\sigma_n)$ is the restriction of Hamiltonian on $V_n$.

We say that the $p$-adic probability distributions (\ref{p*}) are compatible if for all
$n\geq 1$ and $\sigma_{n-1}\in \Phi^{V_{n-1}}$:
\begin{equation}\label{p**}
\sum_{\omega_n\in \Phi^{W_n}}\mu_{\tilde h}^{(n)}(\sigma_{n-1}\vee \omega_n)=\mu_{\tilde h}^{(n-1)}(\sigma_{n-1}).
\end{equation}
Here $\sigma_{n-1}\vee \omega_n$ is the concatenation of the configurations.\\
We note that an analog of the Kolmogorov extension theorem for distributions can be proved for $p$-adic distributions given by (\ref{p*}) (see \cite{16}). According to this theorem there exists a unique $p$-adic measure $\mu_{\tilde h}$ on $\Omega=\Phi^V$ such that,
for all $n$ and $\sigma_n\in \Omega_{V_n}$,
$$\mu_{\tilde h}(\{\sigma|_{V_n}=\sigma_n\})=\mu_{\tilde h}^{(n)}(\sigma_n).$$

Such a measure is called a {\it $p$-adic Gibbs measure} (pGM) corresponding to the Hamiltonian (\ref{ph}) and vector-valued function ${\tilde h}_x, x\in V$.

The following statement describes conditions on ${\tilde h}_x$ guaranteeing compatibility of $\mu_{\tilde h}^{(n)}(\sigma_n)$.

\begin{thm}\label{ep} \label{Theorem1} (see \cite[p.89]{MRasos}) The $p$-adic probability distributions
$\mu_n(\sigma_n)$, $n=1,2,\ldots$, in
(\ref{p*}) are compatible for Potts model iff for any $x\in V\setminus\{x^0\}$
the following equation holds:
\begin{equation}\label{p***}
 h_x=\sum_{y\in S(x)}F(h_y,\theta),
\end{equation}
where $F: h=(h_1, \dots,h_{q-1})\in\mathbb Q_p^{q-1}\to F(h,\theta)=(F_1,\dots,F_{q-1})\in\mathbb Q_p^{q-1}$ is defined as
$$F_i=\log_p\left({(\theta-1)\exp_p(h_i)+\sum_{j=1}^{q-1}\exp_p(h_j)+1\over \theta+ \sum_{j=1}^{q-1}\exp_p(h_j)}\right),$$
$\theta=\exp_p(J)$, $S(x)$ is the set of direct successors of $x$ and $h_x=\left(h_{1,x},\dots,h_{q-1,x}\right)$ with
\begin{equation}\label{hh}
h_{i,x}={\tilde h}_{i,x}-{\tilde h}_{q,x}, \ \ i=1,\dots,q-1.
\end{equation}
\end{thm}

From Theorem \ref{ep} it follows that for any $h=\{h_x,\ \ x\in V\}$
satisfying (\ref{p***}) there exists a unique pGM $\mu_h$ for the $p$-adic Potts model.

\section {translation-invariant $p$-adic gibbs measures for the potts model.}

In this section, we consider $p$-adic Gibbs measures which are translation-invariant, i.e., we assume $h_x=h=(h_1,\dots,h_{q-1})\in\mathbb Q_p^{q-1}$ for all $x\in V$. Then from equation (\ref{p***}) we get $h=kF(h,\theta)$, i.e.,
\begin{equation}\label{pt}
h_i=k\log_p\left({(\theta-1)\exp_p(h_i)+\sum_{j=1}^{q-1}\exp_p(h_j)+1\over \theta+ \sum_{j=1}^{q-1}\exp_p(h_j)}\right),\ \ i=1,\dots,q-1.
\end{equation}
Denoting $z_i=\exp_p(h_i), i=1,\dots,q-1$, we get from (\ref{pt})
\begin{equation}\label{pt1}
z_i=\left({(\theta-1)z_i+\sum_{j=1}^{q-1}z_j+1\over \theta+ \sum_{j=1}^{q-1}z_j}\right)^k,\ \ i=1,\dots,q-1.
\end{equation}
Note that for a solution $z=(z_1,...,z_{q-1})$ of the system of equations (\ref{pt1}) there exists a unique  TIpGMs for the Potts model on the Cayley tree of order $k$ if and only if $z\in\mathcal E_p^{q-1}$.
\begin{thm}\label{tti}  \label{Theorem2} Let $k=2$. Then for any solution $z=(z_1,\dots,z_{q-1})$ of the system of equations (\ref{pt1}) there exists $M\subset \{1,\dots,q-1\}$ and $z^*\in\mathbb Q_p$ such that
$$z_i=\left\{\begin{array}{ll}
1, \ \ \mbox{if} \ \ i\notin M\\[3mm]
z^*, \ \ \mbox{if} \ \ i\in M.
\end{array}
\right.
$$
\end{thm}
\begin{proof} It is easy to see that $z_i=1$ is a solution of $i$th equation of the system (\ref{pt1}) for each $i=1,2,\dots,q-1$. Thus for a given $M\subset \{1,\dots,q-1\}$ one can take $z_i=1$ for any $i\notin M$. Let $\emptyset\ne M\subset\{1,\dots,q-1\}$, without loss of generality we can take $M=\{1,2, \dots, m\}$, $m\leq q-1$, i.e. $z_i=1$, $i=m+1,\dots, q$. Now we shall
prove that $z_1=z_2=\dots=z_m$.  From (\ref{pt1}) we have
\begin{equation}\label{r1}
    z_i=\left(\frac{(\theta-1)z_i+\sum_{j=1}^mz_j+q-m}{\sum_{j=1}^mz_j+q-m-1+\theta}\right)^2, \ \ i=1,\dots,m.
\end{equation}
  By assumption $z_i\neq1, i=1,2,...,m$ from (\ref{r1}) we get
$$(\theta-1)^2=\frac{(z_i-1)\left(\sum_{j=1}^mz_j+q-m+1\right)^2}{z_i^2-z_i}=
\frac{\left(\sum_{j=1}^mz_j+q-m+1\right)^2}{z_i}, \ \ i=1,\dots,m.$$
From these equations we get
$$z_i=z_j\qquad\mbox{for any}\quad i,j\in \{1,\dots,m\}.$$
\end{proof}
By this theorem we have that any TIpGMs of the Potts model on the Cayley tree of order two corresponds to a solution $z^*\in\mathcal E_p$ of the following equation
\begin{equation}\label{rm}
z=f_m(z)\equiv \left({(\theta+m-1)z+q-m\over mz+q-m-1+\theta}\right)^2,
\end{equation}
for some $m=1,\dots,q-1$.
\begin{rk} We note that in the real case Theorem \ref{tti} is true for any $k\geq2$ (see \cite[Theorem 2]{KRK}). But for $p$-adic case if $k\geq3$ then
Theorem \ref{tti} is not true, in general. Indeed

1) If $k=q=p=3$ and $\theta=-2$ then $z=(64,-125)$ is a solution to (\ref{pt1}) and $(64,-125)\in\mathcal E_3^2$.

2) If $k=p=3, q=6$ and $\theta=-\frac{37}{20}$ then $z=(64,-125,1,1,1)\in\mathcal E_3^5$ is a solution to (\ref{pt1}).
\end{rk}
\begin{lemma}\label{l1} If $z(m_1)$ is a solution to (\ref{rm}) with $m=m_1$ then $z^{-1}(m_1)$ is a
solution to (\ref{rm}) with $m=q-m_1$.
\end{lemma}
\begin{proof} It is easy to see that $f_{m}(x)=1/f_{q-m}(x^{-1})$.
\end{proof}

Let $M\subset \{1,\dots, q-1\}$, with $|M|=m$. Then corresponding solution of (\ref{rm})
we denote by $z(M)=\exp_p(h(M))$. It is clear that $h(M)=h(m)$, i.e. it only depends on cardinality of $M$.
Put
$${\mathbf 1}_M=(e_1,\dots,e_q), \ \ \mbox{with} \ \ e_i=1 \ \ \mbox{if} \ \ i\in M, \ \ e_i=0 \ \ \mbox{if} \ \ i\notin M.$$

 We denote by $\mu_{h(M)\mathbf 1_M}$ the TIpGMS corresponding to the solution $h(M)$.
\begin{rk}\label{ns} By formula (\ref{hh}) we have
$${\tilde h}_i(M)=\log_p(\tilde{z}_i(M))=\left\{\begin{array}{ll}
h(M)+{\tilde h}_q(M), \ \ \mbox{if} \ \ i\in M\\[2mm]
{\tilde h}_q(M), \ \ \mbox{if} \ \ i\notin M
\end{array}\right.,$$ i.e.
$${\tilde h}(M){\mathbf 1}_M=h(M){\mathbf 1}_M+{\tilde h}_q(M){\mathbf 1}_{\{1,\dots,q\}}.$$
 Hence for a given $M$, $|M|=m$ and a solution $h(M)$ the number of vectors ${\tilde h}(M){\mathbf 1}_M$ is equal to ${q\choose m}$.
\end{rk}

The following proposition is useful.

\begin{pro}\label{tp} For any finite $\Lambda\subset V$ and any $\sigma_\Lambda\in \{1,\dots,q\}^\Lambda$ we have
\begin{equation}\label{mu}
\mu_{h(M){\mathbf 1}_M}(\sigma_\Lambda)=\mu_{h(M^c){\mathbf 1}_{M^c}}(\sigma_\Lambda),
\end{equation}
where $M^c=\{1,\dots,q\}\setminus M$ and $h(M^c)=-h(M)$.
\end{pro}
Proof is similar to the proof of the Proposition 1 in \cite{KRK}.
The following is a corollary of Theorem \ref{tti} and Proposition \ref{tp}.
\begin{cor}\label{cor1} Each TIpGMs corresponds to a solution of (\ref{rm}) with some $m\leq [q/2]$, where $[a]$ is the integer part of $a$.
Moreover, for a given $m\leq [q/2]$,
a fixed solution to (\ref{rm}) generates ${q\choose m}$ vectors ${\tilde h}$ giving ${q\choose m}$ TIpGMs.
\end{cor}

Now we try to solve the equation \ref{rm} in $\mathcal E_p\setminus\{1\}$. From \ref{rm} we get
$$z-1=\frac{(z-1)(\theta-1)((\theta+2m-1)z+2q-2m+\theta-1)}{(mz+q-m+\theta-1)^2}.$$
Dividing this equation to $z-1$ we obtain
\begin{equation}\label{kv}
m^2z^2+(2m(q-m)-(\theta-1)^2)z+(q-m)^2=0.
\end{equation}
This equation has solutions
\begin{equation}\label{sol}
z_{1,2}(m)=\frac{(\theta-1)^2-2m(q-m)\pm(\theta-1)\sqrt{(\theta-1)^2-4m(q-m)}}{2m^2},
\end{equation}
if there exists $\sqrt{(\theta-1)^2-4m(q-m)}$ in $\mathbb Q_p$. If the equation (\ref{kv}) has solutions $z_{1,2}(m)$ in $\mathbb Q_p$ then we have
\begin{equation}\label{z_1-1z_2-1}
|(z_1(m)-1)(z_2(m)-1)|_p=\frac{\left|q^2-(\theta-1)^2\right|_p}{\left|m^2\right|_p}
\end{equation}
Denote by $D=(\theta-1)^2-4m(q-m)$.
We must check the existence of $\sqrt{D}$ in $\mathbb Q_p$ and $z_{1,2}(m)\in\mathcal E_p\setminus\{1\}$ which equivalent to the following conditions:
\begin{equation}\label{E_p\1}
0<\left|\frac{(\theta-1)^2-2mq\pm(\theta-1)\sqrt{D}}{2m^2}\right|_p<1,\qquad\mbox{if}\ \ p>2
\end{equation}
and
\begin{equation}\label{E_2\1}
0<\left|\frac{(\theta-1)^2-2mq\pm(\theta-1)\sqrt{D}}{2m^2}\right|_2<\frac{1}{2},\qquad\mbox{if}\ \ p=2.
\end{equation}
\begin{rk}\label{z_1=1}
Let $\theta\in\{1-q,1+q\}$ then we have $z_1(m)=1$ and $z_2(m)=\left(\frac{q-m}{m}\right)^2.$ In this case we have only one solution $z_2(m)\neq1$ if $q\neq2m$ and $\left|q^2-2mq\right|_p<\left|2m^2\right|_p$.
\end{rk}

\subsection{\bf Case $p\neq2$}
The following lemma is useful
\begin{lemma}\label{l12}
If $p\neq2$ and $|a|_p=|b|_p$ then $|a+b|_p=|a|_p$ or $|a-b|_p=|a|_p$.
\end{lemma}
\begin{proof}
It is clear that $|a+b|_p=|a-b|_p=|a|_p$ if $a=b=0$. Let $a\neq0$. Since $\left|\frac{a}{|a|_p}\right|_p=1$ for convenience we consider the case $|a|_p=|b|_p=1$. Let us consider the canonical form of $a$ and $b$, i.e.m
$$a=a_0+a_1p+a_2p^2+\cdots,\qquad b=b_0+b_1p+b_2p^2+\cdots,$$
where $a_0, b_0\in\{1,...,p-1\}$.
It is sufficient to show that $a_0+b_0$ or $a_0-b_0$ is not dividable by $p$. Assume $a_0+b_0=p$ then
$a_0-b_0=a_0+b_0-2b_0$. This is not dividable by $p$. Because, $2b_0$ is not dividable by $p$.
\end{proof}
\begin{pro}\label{pro11} Let $p\neq2$. If $q\notin p\mathbb N$ then for any integer number $m\in\{1,...,q-1\}$ the equation (\ref{kv}) has no solution in $\mathcal E_p\setminus\{1\}$.
\end{pro}
\begin{proof}
{\it Case} $|m|_p>\left|(\theta-1)^2\right|_p$. We show that if the solutions (\ref{sol}) exist in $\mathbb Q_p$ then they do not belong to $\mathcal E_p\setminus\{1\}$. Recall that the existence of solutions in $\mathbb Q_p$ is equivalent to the existence of $\sqrt{D}$. Assume that $\sqrt{D}\in\mathbb Q_p$. Then we get
$$\left|(\theta-1)^2D\right|_p=\left|(\theta-1)^2\left((\theta-1)^2-4mq+4m^2\right)\right|_p\leq
\left|(\theta-1)^2m\right|_p<\left|m^2\right|_p.$$
Hence $\left|(\theta-1)\sqrt{D}\right|_p<|m|_p$. Using non-Archimedean norm's property we get
$$|z_{1,2}(m)-1|_p=\frac{\left|(\theta-1)^2-2mq\pm(\theta-1)\sqrt{D}\right|_p}{\left|2m^2\right|_p}=\frac{|m|_p}{\left|m^2\right|_p}=\frac{1}{|m|_p}\geq1.$$
Thus we have shown that the condition (\ref{E_p\1}) is not satisfied.
This means that the solutions do not belong to $\mathcal E_p\setminus\{1\}$.\\
{\it Case} $|m|_p<\left|(\theta-1)^2\right|_p$. Then there exists integer number $s\geq1$ such that
$|m|_p=\left|p^s(\theta-1)^2\right|_p$. We have
$$D=(\theta-1)^2\left(1+\varepsilon p^s\right),\qquad\mbox{where}\ |\varepsilon|_p=1.$$
By Theorem \ref{tx2} there exists $\sqrt{D}$ and $\sqrt{D}=(\theta-1)(1+\varepsilon' p^s)$. Consequently, the solutions (\ref{sol}) exist in $\mathbb Q_p$. Now we shall show that $z_{1,2}(m)\notin\mathcal E_p\setminus\{1\}$. We have from (\ref{z_1-1z_2-1})
$$|z_1(m)-1|_p=\frac{\left|(\theta-1)^2-2mq+(\theta-1)^2\left(1+\varepsilon' p^s\right)\right|_p}{\left|2m^2\right|_p}=\frac{\left|(\theta-1)^2\right|_p}{\left|m^2\right|_p}>\frac{1}{|m|_p}>1.$$
From this and by (\ref{z_1-1z_2-1}) we get
$$|z_2(m)-1|_p=\frac{1}{\left|m^2\right|_p}\cdot\frac{\left|m^2\right|_p}{\left|(\theta-1)^2\right|_p}=\frac{1}{\left|(\theta-1)^2\right|_p}>1.$$
This means that $z_{1,2}(m)\notin\mathcal E_p\setminus\{1\}$.\\
{\it Case} $|m|_p=\left|(\theta-1)^2\right|_p$. If $\left|(\theta-1)^2-4mq\right|_p<\left|(\theta-1)^2\right|_p$ then from non-Archimedean norm's property we get
$$\left|(\theta-1)^2-2mq\right|_p=\left|(\theta-1)^2-4mq+2mq\right|_p=|m|_p=\left|(\theta-1)^2\right|_p.$$
Consequently
$$|z_{1,2}(m)-1|_p=\frac{\left|(\theta-1)^2-2mq\pm(\theta-1)\sqrt{D}\right|_p}{\left|2m^2\right|_p}=\frac{\left|(\theta-1)^2\right|_p}{\left|m^2\right|_p}=\frac{1}{|m|_p}>1.$$
Now let $\left|(\theta-1)^2-2mq\right|_p<\left|(\theta-1)^2\right|_p$. Then from non-Archimedean norm's property we get
$$\left|(\theta-1)^2-4mq\right|_p=\left|(\theta-1)^2-2mq-2mq\right|_p=\left|(\theta-1)^2\right|_p.$$
Hence
$$|z_{1,2}(m)-1|_p=\frac{\left|(\theta-1)^2-2mq\pm(\theta-1)\sqrt{D}\right|_p}{\left|2m^2\right|_p}=\frac{\left|(\theta-1)^2\right|_p}{\left|m^2\right|_p}=\frac{1}{|m|_p}>1.$$ Finally we consider the case
$$\left|(\theta-1)^2-2mq\right|_p=\left|(\theta-1)^2-4mq\right|_p=\left|(\theta-1)^2\right|_p.$$
If $\sqrt{D}$ exists in $\mathbb Q_p$ then we have $|\sqrt{D}|_p=|\theta-1|_p$. There exist $p$-adic numbers $\varepsilon$ and $\epsilon$ such that
$$(\theta-1)^2-2mq=(\theta-1)^2\varepsilon,\qquad (\theta-1)\sqrt{D}=(\theta-1)^2\epsilon\qquad\mbox{and}\ |\varepsilon|_p=|\epsilon|_p=1.$$
By Lemma \ref{l12} we get $|\varepsilon+\epsilon|_p=1$ or $|\varepsilon-\epsilon|_p=1$ as $p\neq2$. Assume that $|\varepsilon+\epsilon|_p=1$ (The case $|\varepsilon-\epsilon|_p=1$ is similar). Then for the solution $z_{1}(m)$ we get
$$|z_{1}(m)-1|_p=\frac{\left|(\theta-1)^2\left(\varepsilon+\epsilon\right)\right|_p}{\left|m^2\right|_p}=\frac{\left|(\theta-1)^2\right|_p}{\left|m^2\right|_p}=\frac{1}{|m|_p}>1.$$
From this and by (\ref{z_1-1z_2-1}) we get
$$|z_2(m)-1|_p=\frac{1}{\left|m^2\right|_p}\cdot|m|_p=\frac{1}{|m|_p}>1.$$
This means that the solutions do not belong to $\mathcal E_p\setminus\{1\}$.
\end{proof}
\begin{pro}\label{pro12}
Let $p\neq2,\ q\in p\mathbb N$ and $\theta\in\{1-q,1+q\}$. Then the following statements hold\\
1) If $|m|_p>|q|_p$ then the equation (\ref{kv}) has only one solution $z_2(m)$ in $\mathcal E_p\setminus\{1\}$.\\
2) If $|m|_p<|q|_p$ then the equation (\ref{kv}) has no solution in $\mathcal E_p\setminus\{1\}$.\\
3) If $|m|_p=|q|_p$ and $|q-2m|_p\in\left\{0,|q|_p\right\}$ then the equation (\ref{kv}) has no solution in $\mathcal E_p\setminus\{1\}$.\\
4) If $|m|_p=|q|_p$ and $0<|q-2m|_p<|q|_p$ then the equation (\ref{kv}) has only one solution $z_2(m)$ in $\mathcal E_p\setminus\{1\}$.
\end{pro}
\begin{proof}
By Remark \ref{z_1=1} we have $z_1(m)=1\notin\mathcal E_p\setminus\{1\}$ and $z_2(m)=\left(\frac{q-m}{m}\right)^2$. It easy to see that $z_2(m)\in\mathcal E_p\setminus\{1\}$ is equivalent to the condition
\begin{equation}\label{theta=1+q}
0<\left|q^2-2mq\right|_p<\left|m^2\right|_p
\end{equation}
So, we must check condition (\ref{theta=1+q}).\\
Let $|q|_p\neq|m|_p$. Then by non-Archimedean norm's property we get
$$\left|q^2-2mq\right|_p<\left|m^2\right|_p,\qquad\mbox{if}\quad |m|_p>|q|_p$$
and
$$\left|q^2-2mq\right|_p>\left|m^2\right|_p,\qquad\mbox{if}\quad |m|_p<|q|_p.$$
Let $|q|_p=|m|_p$. It is easy to see condition (\ref{theta=1+q}) is not satisfied if $q=2m$. If $|q-2m|_p=|q|_p$ then we have
$$\left|q^2-2mq\right|_p=|q(q-2m)|_p=\left|q^2\right|_p=\left|m^2\right|_p.$$
If $0<|q-2m|_p<|q|_p$ then we have
$$\left|q^2-2mq\right|_p=|q(q-2m)|_p<\left|q^2\right|_p=\left|m^2\right|_p.$$
\end{proof}
\begin{pro}\label{pro13}
Let $p\neq2,\ q\in p\mathbb N$ and $\theta\notin\{1-q,1+q\}$.\\
1) If $|m|_p>\max\{|\theta-1|_p,|q|_p\}$  then there exist two solutions in $\mathcal E_p\setminus\{1\}$\\
2) If $|\theta-1|_p>\max\{|m|_p,|q|_p\}$ then the equation (\ref{kv}) has no solutions in $\mathcal E_p\setminus\{1\}$.\\
3) If $|q|_p>\max\{|m|_p,|\theta-1|_p\}$ then the equation (\ref{kv}) has no solutions in $\mathcal E_p\setminus\{1\}$.\\
4) If $|q|_p<|m|_p=|\theta-1|_p$ then the equation (\ref{kv}) has no solutions in $\mathcal E_p\setminus\{1\}$.\\
5) If $|\theta-1|_p<|q|_p=|m|_p$ then the equation (\ref{kv}) has no solutions in $\mathcal E_p\setminus\{1\}$.\\
6) If $|m|_p<|\theta-1|_p=|q|_p$ and $\left|(\theta-1)^2-q^2\right|_p<\left|q^2\right|_p$ then the equation (\ref{kv}) has only one solution $z_2(m)$ in $\mathcal E_p\setminus\{1\}$.\\
7) If $|m|_p<|\theta-1|_p=|q|_p$ and $\left|(\theta-1)^2-q^2\right|_p=\left|q^2\right|_p$ then the equation (\ref{kv}) has no solution in $\mathcal E_p\setminus\{1\}$.\\
8) Let $|m|_p=|\theta-1|_p=|q|_p$. If $\left|(\theta-1)^2-q^2\right|_p=\left|q^2\right|_p$ then the equation (\ref{kv}) has no solution in $\mathcal E_p\setminus\{1\}$.\\
9) Let $|m|_p=|\theta-1|_p=|q|_p$. If $|\theta-1+q|_p<|q|_p$ $(|\theta-1-q|_p<|q|_p)$ and $|q-2m|_p=|q|_p$ then the equation (\ref{kv}) has only one solution $z_1(m)$ (resp. $z_2(m)$) in $\mathcal E_p\setminus\{1\}$.\\
10) Let $|m|_p=|\theta-1|_p=|q|_p$ and $\left|(\theta-1)^2-q^2\right|_p<\left|q^2\right|_p,\ |q-2m|_p<|q|_p$. Then the equation has two solutions in $\mathcal E_p\setminus\{1\}$ iff $\sqrt{D}$ exists in $\mathbb Q_p$.
\end{pro}
\begin{proof} Note that if there exist $z_{1,2}(m)$ in $\mathbb Q_p$ then from $\theta\neq1\pm q$ we have $z_{1,2}(m)\neq1$. So, instead of (\ref{E_p\1}) we must check the following
\begin{equation}\label{E_p}
|z_{1,2}(m)-1|_p=\left|\frac{(\theta-1)^2-2mq\pm(\theta-1)\sqrt{D}}{2m^2}\right|_p<1
\end{equation}
1) Let $|m|_p>\max\{|\theta-1|_p,|q|_p\}$. In this case we have
$$D=(\theta-1)^2-4mq+4m^2=4m^2\left(1-\frac{q}{m}+\left(\frac{\theta-1}{2m}\right)^2\right)=4m^2(1+\varepsilon p),\qquad |\varepsilon|_p\leq1.$$
By Theorem \ref{tx2} there exists $\sqrt{D}$ in $\mathbb Q_p$ and $\sqrt{D}=2m(1+\varepsilon'p)$ where $|\varepsilon'|_p\leq1$. Consequently, the equation (\ref{kv}) has two solutions $z_1(m)$ and $z_2(m)$ in $\mathbb Q_p$. We shall check (\ref{E_p}). From non-Archimedean norm's property we get
$$|z_{1,2}-1|_p=\left|\frac{(\theta-1)^2-2mq\pm2m(\theta-1)(1+\varepsilon' p)}{2m^2}\right|_p\leq\frac{\max\{|\theta-1|_p,|q|_p\}}{|m|_p}<1.$$
Hence, $z_{1,2}\in\mathcal E_p\setminus\{1\}$.

2) Let $|\theta-1|_p>\max\{|m|_p,|q|_p\}$. In this case for the discriminant we get
$$D=(\theta-1)^2\left(1+\frac{4m(m-q)}{(\theta-1)^2}\right)=(\theta-1)^2(1+\varepsilon p)\qquad\mbox{where}\ |\varepsilon|_p\leq1.$$
By Theorem \ref{tx2} there exists $\sqrt{D}$ and $\sqrt{D}=(\theta-1)(1+\varepsilon' p)$ where $|\varepsilon|_p\leq1$. Consequently,
the equation (\ref{kv}) has two solutions in $\mathbb Q_p$. For the solution $z_1(m)$ from (\ref{E_p}) we get
$$|z_1(m)-1|_p=\frac{\left|2(\theta-1)^2-2mq+\varepsilon' p(\theta-1)^2\right|_p}{\left|m^2\right|_p}=\frac{\left|(\theta-1)^2\right|_p}{\left|m^2\right|_p}>1.$$
From this and by (\ref{z_1-1z_2-1}) we have
$$|z_2(m)-1|_p=\frac{\left|(\theta-1)^2\right|_p}{\left|m^2\right|_p}\cdot\frac{\left|m^2\right|_p}{\left|(\theta-1)^2\right|_p}=1.$$
This means that the solutions $z_{1,2}(m)$ do not belong to the set $\mathcal E_p$.

3) Let $|\theta-1|_p\leq|m|_p<|q|_p$. In this case the equation (\ref{kv}) is solvable in $\mathbb Q_p$ if and only if $\sqrt{-mq}$ exists in $\mathbb Q_p$. Assume that $\sqrt{-mq}\in\mathbb Q_p$. Since $\left|(\theta-1)^2\right|_p<|mq|_p$ and
$$\left|(\theta-1)^2D\right|_p=\left|-4mq(\theta-1)^2(1+\varepsilon p)\right|_p<\left|m^2q^2\right|_p\qquad\mbox{where}\ |\varepsilon|_p\leq1$$
by non-Archimedean norm's property we have
$$|z_{1,2}(m)-1|_p=\left|\frac{(\theta-1)^2-2mq\pm(\theta-1)\sqrt{D}}{2m^2}\right|_p=
\frac{\left|mq\right|_p}{\left|m^2\right|_p}=\frac{|q|_p}{|m|_p}>1.$$
It means that in this case the equation (\ref{kv}) has no solution in $\mathcal E_p$.

Let $|m|_p<|\theta-1|_p<|q|_p$. If $\left|(\theta-1)^2\right|_p>|mq|_p$ then by Theorem \ref{tx2} there exists $\sqrt{D}$ in $\mathbb Q_p$ and $\sqrt{D}=(\theta-1)(1+\varepsilon p)$. For the solution $z_1(m)$ from (\ref{E_p}) we get
$$|z_1(m)-1|_p=\frac{\left|2(\theta-1)^2-2mq+\varepsilon p(\theta-1)^2\right|_p}{\left|m^2\right|_p}=\frac{\left|(\theta-1)^2\right|_p}{\left|m^2\right|_p}>1.$$
By substituting this to (\ref{z_1-1z_2-1}) we have
$$|z_2(m)-1|_p=\frac{\left|q^2\right|_p}{\left|m^2\right|_p}\cdot\frac{\left|m^2\right|_p}{\left|(\theta-1)^2\right|_p}=
\frac{\left|q^2\right|_p}{\left|(\theta-1)^2\right|_p}>1.$$

If $\left|(\theta-1)^2\right|_p\leq|mq|_p$ then we have
$$\left|(\theta-1)^2D\right|_p=\left|(\theta-1)^2\left((\theta-1)^2-4mq+4m^2\right)\right|_p\leq\left|m^2q^2\right|_p.$$
Considering this by non-Archimedean norm's property from (\ref{E_p}) we get
$$|z_{1,2}-1|_p\leq\frac{\left|mq\right|_p}{\left|m^2\right|_p}=\frac{|q|_p}{|m|_p}.$$
On the other hand by (\ref{z_1-1z_2-1}) we have $|(z_1(m)-1)(z_2(m)-1)|_p=\frac{\left|q^2\right|_p}{\left|m^2\right|_p}>1$.\\
Consequently,
$$|z_{1,2}-1|_p=\frac{|q|_p}{|m|_p}>1.$$

4) Let $|q|_p<|m|_p=|\theta-1|_p$. Now we shall prove that if there exists $\sqrt{\frac{(\theta-1)^2}{m^2}+4-\frac{4q}{m}}$ in $\mathbb Q_p$ then holds
\begin{equation}\label{e}
\left|\frac{\theta-1}{m}\pm\sqrt{\frac{(\theta-1)^2}{m^2}+4-\frac{4q}{m}}\right|_p=1
\end{equation}
Assume that $\sqrt{\frac{(\theta-1)^2}{m^2}+4-\frac{4q}{m}}$ exists. Then by $\frac{|q|_p}{|m|_p}<1$ we get
\begin{equation}\label{e1}
\left|\left(\frac{\theta-1}{m}+\sqrt{\frac{(\theta-1)^2}{m^2}+4-\frac{4q}{m^2}}\right)
\left(\frac{\theta-1}{m}-\sqrt{\frac{(\theta-1)^2}{m^2}+4-\frac{4q}{m}}\right)\right|_p=1.
\end{equation}
Considering $\frac{|\theta-1|_p}{|m|_p}=1$ we have
\begin{equation}\label{e2}
\left|\frac{\theta-1}{m}\pm\sqrt{\frac{(\theta-1)^2}{m^2}+4-\frac{4q}{m}}\right|_p\leq1
\end{equation}
From (\ref{e1}),(\ref{e2}) it follows (\ref{e}).\\
Since $D=m^2\left(\frac{(\theta-1)^2}{m^2}+4-\frac{4q}{m}\right)$ there exists $\sqrt{D}$ if and only if
$\sqrt{\frac{(\theta-1)^2}{m^2}+4-\frac{4q}{m}}$ exists.
If $\sqrt{\frac{(\theta-1)^2}{m^2}+4-\frac{4q}{m}}$ exists then from $|q|_p<|m|_p=|\theta-1|_p$ and by (\ref{e}) we have
$$|z_{1,2}(m)-1|_p=\left|\frac{\theta-1}{2m}\left(\frac{\theta-1}{m}\pm\sqrt{\frac{(\theta-1)^2}{m^2}+4-\frac{4q}{m^2}}\right)-\frac{q}{m}\right|_p=1.$$ This means that $z_{1,2}(m)\notin\mathcal E_p$.

5) Let $|\theta-1|_p<|m|_p=|q|_p$. In this case if a discriminant $\sqrt{D}$ exists then we have following inequality
$$\left|(\theta-1)^2D\right|_p=\left|(\theta-1)^2\left((\theta-1)^2-4m(q-m)\right)\right|_p\leq\left|(\theta-1)^2m^2\right|_p<\left|m^4\right|_p.$$
Hence,
$$|z_{1,2}(m)-1|_p=\frac{\left|(\theta-1)^2-2mq\pm(\theta-1)\sqrt{D}\right|_p}{\left|2m^2\right|_p}=\frac{\left|m^2\right|_p}{\left|m^2\right|_p}=1.$$
It means that $z_{1,2}(m)\notin\mathcal E_p$.

6) Let $|m|_p<|\theta-1|_p=|q|_p$ and $0<\left|(\theta-1)^2-q^2\right|_p<\left|q^2\right|_p$. In this case $\sqrt{D}$ exists and $\sqrt{D}=(\theta-1)(1+\varepsilon p)$ where $|\varepsilon|_p\leq1$. Consequently, the equation (\ref{kv}) has solutions $z_{1,2}(m)$ in $\mathbb Q_p$. From (\ref{E_p}) and(\ref{z_1-1z_2-1}) we have
$$|z_1(m)-1|_p=\frac{\left|(\theta-1)^2\right|_p}{|m^2|}>1.$$
It means that $z_{1}(m)\notin\mathcal E_o\setminus\{1\}$.\\
By (\ref{z_1-1z_2-1}) we get
$$|(z_1(m)-1)(z_2(m)-1)|_p=\frac{\left|(\theta-1)^2-q^2\right|}{|m^2|_p}<\frac{\left|(\theta-1)^2\right|_p}{\left|m^2\right|_p}.$$
From these we get $|z_2(m)-1|_p<1$. It means $z_2(m)\in\mathcal E_p\setminus\{1\}.$

7) Let $|m|_p<|\theta-1|_p=|q|_p$ and $\left|(\theta-1)^2-q^2\right|_p=\left|q^2\right|_p$. In this case the equation (\ref{kv}) has two solutions in $\mathbb Q_p$. We show that they do not belong to $\mathcal E_p\setminus\{1\}$. By (\ref{E_p}) we have
$$|z_1(m)-1|_p=\frac{\left|(\theta-1)^2\right|_p}{\left|m^2\right|_p}>1$$
Hence, by (\ref{z_1-1z_2-1}) we have
$$|(z_1(m)-1)(z_2(m)-1)|_p=\frac{\left|(\theta-1)^2-q^2\right|}{|m^2|_p}=\frac{\left|(\theta-1)^2\right|_p}{\left|m^2\right|_p}.$$
From these we get $|z_2(m)-1|_p=1$. Thus we have shown that $z_{1,2}(m)\notin\mathcal E_p\setminus\{1\}$.

8) Let $|m|_p=|\theta-1|_p=|q|_p$ and $\left|(\theta-1)^2-q^2\right|_p=\left|q^2\right|_p$. If there exist solutions to the equation (\ref{kv}) then (\ref{E_p}) we have
$$|z_{1,2}(m)-1|_p=\frac{\left|(\theta-1)^2-2mq\pm(\theta-1)\sqrt{D}\right|_p}{\left|2m^2\right|_p}\leq\frac{\left|q^2\right|_p}{\left|m^2\right|_p}=1.$$
But from (\ref{z_1-1z_2-1}) we get
$$|(z_1(m)-1)(z_2(m)-1)|_p=\frac{\left|(\theta-1)-q^2\right|_p}{\left|m^2\right|_p}=\frac{\left|q^2\right|_p}{\left|m^2\right|_p}=1.$$
Consequently,
$$|z_{1,2}(m)-1|_p=1.$$

9) Let $|m|_p=|\theta-1|_p=|q|_p$ and $|\theta-1+q|_p<|q|_p,\ |q-2m|_p=|q|_p$. Then by Lemma \ref{l12} we have $|\theta-1-q|_p=|q|_p$ and $p>2$. In this case we get
$$D=(\theta-1)^2-q^2+(q-2m)^2=(q-2m)^2(1+\varepsilon p).$$
By Theorem \ref{tx2} there exists $\sqrt{D}$ in $\mathbb Q_p$ and $\sqrt{D}=(q-2m)(1+\varepsilon' p)$.\\
From (\ref{E_p}) we get
$$|z_1(m)-1|_p=\frac{\left|(\theta-1)^2-q^2-(q-2m)(\theta-1+q)+\varepsilon' p(q-2m)(\theta-1)\right|_p}{\left|2m^2\right|_p}<\frac{\left|q^2\right|_p}{\left|m^2\right|_p}=1$$
and
$$|z_2(m)-1|_p=\frac{\left|(\theta-1)^2-q^2-(q-2m)(\theta-1-q)-\varepsilon' p(q-2m)(\theta-1)\right|_p}{\left|2m^2\right|_p}=\frac{\left|q^2\right|_p}{\left|m^2\right|_p}=1.$$
Thus we have shown that $z_1(m)\in\mathcal E_p\setminus\{1\}$ and $z_2(m)\notin\mathcal E_p\setminus\{1\}$.

10) Let $|m|_p=|\theta-1|_p=|q|_p$. If there exists $\sqrt{D}$ in $\mathbb Q_p$ then from $\left|(\theta-1)^2-q^2\right|_p<\left|q^2\right|_p$ and $|q-2m|_p<|q|_p$ we get
$$|D|_p=\left|(\theta-1)^2-q^2+(q-2m)^2\right|_p<\left|q^2\right|_p.$$
Hence,
$$|z_{1,2}(m)-1|_p=\frac{\left|(\theta-1)^2-q^2+q(q-2m)\pm(\theta-1)\sqrt{D}\right|_p}{\left|m^2\right|_p}<\frac{\left|q^2\right|_p}{\left|m^2\right|_p}=1.$$
Thus we have shown that in this case the equation (\ref{kv}) has two solutions in $\mathcal E_p\setminus\{1\}$.
\end{proof}

\begin{cor}\label{cor11}
Let $p\neq2$ and $q\in p\mathbb N$.\\
a) If $|m|_p=1$ and $\theta\notin\{1-q,1+q\}$ then the equation (\ref{kv}) has two solutions $z_1(m)$ and $z_2(m)$ in $\mathcal E_p\setminus\{1\}$.\\
b) If $|m|_p=1$ and $\theta\in\{1-q,1+q\}$ then the equation (\ref{kv}) has only one solution $z_2(m)$ in $\mathcal E_p\setminus\{1\}$.
\end{cor}
\begin{cor}\label{cor12}
Let $q=p>2$.\\
1) If $\theta\notin\{1-q,1+q\}$ then for any integer number $m$ such that $m<q$ the equation (\ref{kv}) has two solutions $z_1(m)$ and $z_2(m)$ in $\mathcal E_p\setminus\{1\}$;\\
2) If $\theta\in\{1-q,1+q\}$ then for any integer number $m$ such that $m<q$ the equation (\ref{kv}) has only one solution $z_2(m)$ in $\mathcal E_p\setminus\{1\}$.
\end{cor}
By Corollary \ref{cor1} and by Propositions \ref{pro11}-\ref{pro13} we get the following
\begin{thm}\label{tigm1} Let $p\neq2$.
1) For a given $m\leq[q/2]$ there exist 2${q\choose m}$ of TIpGMs if at least one of the following conditions is satisfied\\
1a) $|m|_p>\max\{|\theta-1|_p,|q|_p\}$ and $\theta\notin\{1-q,1+q\}$\\
1b) $|m|_p=|\theta-1|_p=|q|_p,\ 0<\left|(\theta-1)^2-q^2\right|_p<\left|q^2\right|_p,\ 0<|q-2m|_p<|q|_p$ and there exists an integer number $s\geq1$ such that $\left|p^{-2s}\left((\theta-1)^2-4m(q-m)\right)\right|_p=1$.\\
2) For a given $m\leq[q/2]$ there exist ${q\choose m}$ of TIpGMs if at least one of the following conditions is satisfied\\
2a) $|m|_p>\max\{|\theta-1|_p,|q|_p\}$ and $\theta\in\{1-q,1+q\}$\\
2b) $|m|_p<|\theta-1|_p=|q|_p$ and $0<\left|(\theta-1)^2-q^2\right|_p<\left|q^2\right|_p$\\
2c) $|m|_p=|\theta-1|_p=|q|_p$ and $0<\left|(\theta-1)^2-q^2\right|_p<\left|q^2\right|_p,\ |q-2m|_p=|q|_p$\\
2d) $|m|_p=|\theta-1|_p=|q|_p$ and $\theta\in\{1-q,1+q\}$ and $0<|q-2m|_p<|q|_p$\\
2e) $q=2m$, $|\theta-1|_p=|q|_p$, $0<\left|(\theta-1)^2-q^2\right|_p<\left|q^2\right|_p$ and there exists an integer number $s\geq1$ such that $\left|p^{-2s}\left((\theta-1)^2-q^2\right)\right|_p=1$.\\
Otherwise for a given $m\in \{1,\dots, [q/2]\}$ there does not exist any TIpGM.
\end{thm}

\subsection{\bf Case $p=2$}
\begin{pro}\label{pro21} Let $p=2$. If $|q|_2>\frac{1}{4}$ then for any integer number $m<q$ the equation (\ref{kv}) has no solution in $\mathcal E_2\setminus\{1\}$.
\end{pro}
\begin{proof}
{\it Case} $|q|_2=1$. Let $\left|(\theta-1)^2\right|_2<|2m|_2$. Then we have
$$\left|(\theta-1)^2D\right|_2=\left|(\theta-1)^2\left((\theta-1)^2-4mq+4m^2\right)\right|_2\leq\left|(\theta-1)^24m\right|_2<\left|4m^2\right|_2.$$
From this and by non-Archimedean norm's property we get
$$|z_{1,2}(m)-1|_2=\frac{\left|(\theta-1)^2-2mq\pm(\theta-1)\sqrt{D}\right|_2}{\left|2m^2\right|_2}=\frac{\left|2m\right|_2}{\left|2m^2\right|_2}=\frac{1}{|m|_2}\geq1.$$
Let $\left|(\theta-1)^2\right|_2>|2m|_2$. It is easy to see that by Theorem \ref{tx2} there does not exist $\sqrt{D}$ in $\mathbb Q_2$ if $\left|(\theta-1)^2\right|_2=|m|_2$.\\
Assume that $\left|(\theta-1)^2\right|_2>|m|_2$. Then by Theorem \ref{tx2} there exists $\sqrt{D}$ and $\sqrt{D}=(\theta-1)(1+2\varepsilon)$ where $|\varepsilon|_2\leq1$.\\
If $|\varepsilon|_2<1$ then for the solution $z_1(m)$ we have
$$|z_1(m)-1|_2=\frac{\left|2(\theta-1)^2-2mq+2(\theta-1)^2\varepsilon\right|_2}{\left|2m^2\right|_2}=\frac{\left|(\theta-1)^2\right|_2}{\left|m^2\right|_2}>1$$
From this and by (\ref{z_1-1z_2-1}) we get
$$|z_2(m)-1|_2=\frac{1}{\left|m^2\right|_2}\cdot\frac{\left|m^2\right|_2}{\left|(\theta-1)^2\right|_2}>1.$$
If $|\varepsilon|_2=1$ then for the solution $z_2(m)$ we have
$$|z_2(m)-1|_2=\frac{\left|-2mq-2(\theta-1)^2\varepsilon\right|_2}{\left|2m^2\right|_2}=\frac{\left|(\theta-1)^2\right|_2}{\left|m^2\right|_2}>1$$
From this and by (\ref{z_1-1z_2-1}) we get
$$|z_1(m)-1|_2=\frac{1}{\left|m^2\right|_2}\cdot\frac{\left|m^2\right|_2}{\left|(\theta-1)^2\right|_2}>1.$$
Let $\left|(\theta-1)^2\right|_2=|2m|_2$. In this case we have
$$\left|(\theta-1)^2D\right|_2=\left|(\theta-1)^22m\right|_2=\left|4m^2\right|_2.$$
Note that if $|a|_2=|b|_2=|c|_2$ then follows $|a\pm b\pm c|_2=|a|_2$.
From this property we get
$$|z_{1,2}(m)-1|_2=\frac{\left|(\theta-1)^2-2mq\pm(\theta-1)\sqrt{D}\right|_2}{\left|m^2\right|_2}=\frac{\left|2m\right|_2}{\left|m^2\right|_2}=\frac{1}{|m|_2}>1.$$
Thus we have shown that the equation (\ref{kv}) has no solution in $\mathcal E_2\setminus\{1\}$ if $|q|_2=1$.\\
{\it Case} $|q|_2=\frac{1}{2}$. If $|m|_2=1$ then for the discriminant we have $D=4m^2(1+2\varepsilon)$ where $|\varepsilon|_2=1$. By Theorem \ref{tx2} there does not exist $\sqrt{D}$ in $\mathbb Q_2$.\\
If $|m|_2=|q|_2$ then we get
$$\left|(\theta-1)^2D\right|_2\leq\left|16(\theta-1)^2\right|_2.$$
Considering $\left|\theta-1\right|_2\leq\frac{1}{4}$ and $|2mq|_2=\frac{1}{8}$ we have
$$|z_{1,2}(m)-1|_2=\frac{|2mq|_2}{\left|2m^2\right|_2}=1.$$
Let $|m|_2<|q|_2$. If $\left|(\theta-1)^2\right|_2\leq|8m|_2$ then we have
$$\left|(\theta-1)^2D\right|_2\leq\left|8m(\theta-1)^2\right|_2\leq\left|64m^2\right|_2.$$
Hence,
$$|z_{1,2}(m)-1|_2=\frac{|2mq|_2}{\left|2m^2\right|_2}=\frac{|q|_2}{|m|_2}>1.$$
If $\left|(\theta-1)^2\right|_2>|8m|_2$ then by Theorem \ref{tx2} there exists $\sqrt{D}$ if and only if
$\left|(\theta-1)^2\right|_2\geq|m|_2$. Assume that $\left|(\theta-1)^2\right|_2\geq|m|_2$. Then we have
$$|z_1(m)-1|_2=\frac{\left|(\theta-1)^2-2mq+(\theta-1)^2(1+2\varepsilon)\right|_2}{\left|2m^2\right|_2}=
\frac{\left|(\theta-1)^2\right|_2}{\left|m^2\right|_2}\geq1.$$
Considering $|q|_2=\frac{1}{2}$ and $|\theta-1|_2<\frac{1}{2}$ from (\ref{z_1-1z_2-1}) we get
$$|z_2(m)-1|_2=\frac{\left|q^2\right|_2}{\left|m^2\right|_2}\cdot\frac{\left|m^2\right|_2}{\left|(\theta-1)^2\right|_2}>1.$$
Thus we have proved that the equation (\ref{kv}) has no solution in $\mathcal E_2\setminus\{1\}$ if $|q|_2=\frac{1}{2}$.
\end{proof}
\begin{pro}\label{pro22} Let $p=2$ and $\theta\in\{1-q,\ 1+q\}$. Then the equation (\ref{kv}) has only solution $z_2(m)$ in $\mathcal E_2\setminus\{1\}$ if $|m|_2>|q|_2$, otherwise it has no solution in $\mathcal E_2\setminus\{1\}$.
\end{pro}
\begin{proof}
Let $\theta=1\pm q$. Then by Remark \ref{theta=1+q} we get $z_1(m)=1$ and $z_2(m)=\left(\frac{q-m}{m}\right)^2$. It is clear that $z_1(m)\notin\mathcal E_2\setminus\{1\}$.
\[|z_2(m)-1|_2=\frac{\left|q^2-2mq\right|_2}{\left|m^2\right|_2}=\left\{\begin{array}{ll}
>\frac{1}{2}&\mbox{ if } |m|_2\leq|q|_2\\
0&\mbox{ if } q=2m\\
<\frac{1}{2}&\mbox{ if } |m|_2>|q|_2
\end{array}\right.\]
Hence, the equation (\ref{kv}) has solution in $\mathcal E_2\setminus\{1\}$ if and only if $|q|_2<|m|_2$.
\end{proof}
\begin{pro}\label{pro23} Let $p=2$ and $\theta\notin\{1-q,1+q\}$. Then the following statements hold\\
1) If $|4m|_2>\max\{|\theta-1|_2,|q|_2\}$ then the equation (\ref{kv}) has two solutions $z_1(m)$ and $z_2(m)$ in $\mathcal E_2\setminus\{1\}$\\
2) If $|\theta-1|_2>\max\{|q|_2,|4m|_2\}$ then the equation (\ref{kv}) has no solution in $\mathcal E_2\setminus\{1\}$.\\
3) If $|q|_2>\max\{|\theta-1|_2,|4m|_2\}$ then the equation (\ref{kv}) has no solution in $\mathcal E_2\setminus\{1\}$.\\
4) If $|4m|_2=|\theta-1|_2>|q|_2$ then then the equation (\ref{kv}) has no solution in $\mathcal E_2\setminus\{1\}$.\\
5) If $|4m|_2=|q|_2>|\theta-1|_2$ then then the equation (\ref{kv}) has no solution in $\mathcal E_2\setminus\{1\}$.\\
6) If $|4m|_2=|\theta-1|_2=|q|_2$ then the equation (\ref{kv}) has two solutions $z_1(m)$ and $z_2(m)$ in $\mathcal E_2\setminus\{1\}$\\
7) If $|m|_2=|\theta-1|_2=|q|_2>|4m|_2$ then the equation (\ref{kv}) has only one solution in $\mathcal E_2\setminus\{1\}$\\
8) Let $|m|_2>|\theta-1|_2=|q|_2>|4m|_2$. If there exists $\sqrt{D}$ then the equation (\ref{kv}) has two solutions $z_1(m)$ and $z_2(m)$ in $\mathcal E_2\setminus\{1\}$.\\
9) If $|\theta-1|_2=|q|_2>|m|_2$ then the equation (\ref{kv}) has only one solution in $\mathcal E_2\setminus\{1\}$.
\end{pro}
\begin{proof}
1) Let $|4m|_2>\max\{|\theta-1|_2,|q|_2\}$. Then it is clear that by Theorem \ref{tx2} there exists $\sqrt{D}$ and $\sqrt{D}=2m(1+2\varepsilon)$. Then for the solutions $z_{1,2}(m)$ we get
$$|z_{1,2}(m)-1|_2=\frac{\left|(\theta-1)^2-2mq\pm2m(\theta-1)(1+2\varepsilon)\right|_2}{\left|m^2\right|_2}\leq\frac{\max\{|q|_2,|\theta-1|_2\}}{|m|_2}<\frac{1}{4}.$$
It means that $z_{1,2}\in\mathcal E_2\setminus\{1\}$.

2) Let $|\theta-1|_2>\max\{|q|_2,|4m|_2\}$. Then we have
$$\left|D\right|_2=\left|(\theta-1)^2+4m^2-4mq\right|_2<\left|(\theta-1)^2\right|_2.$$
Hence
$$|z_{1,2}(m)-1|_2=\frac{\left|(\theta-1)^2\right|_2}{\left|2m^2\right|_2}=\frac{1}{2}$$
It means $z_{1,2}(m)\notin\mathcal E_2\setminus\{1\}$.\\
It easy to see that there does not exist $\sqrt{D}$ if $|\theta-1|_2=|m|_2$.\\
If $|\theta-1|_2>|m|_2$ then by Theorem \ref{tx2} there exists $\sqrt{D}$ and $\sqrt{D}=(\theta-1)(1+2\varepsilon)$. From non-Archimedean norm's property we get
$$|z_1(m)-1|_2=\frac{\left|(\theta-1)^2\right|_2}{\left|m^2\right|_2}>1,\qquad |z_2(m)-1|_2=1\qquad\mbox{if}\ |\varepsilon|_2<1$$
and
$$|z_1(m)-1|_2=1,\qquad |z_2(m)-1|_2=\frac{\left|(\theta-1)^2\right|_2}{\left|m^2\right|_2}>1\qquad\mbox{if}\ |\varepsilon|_2=1.$$

3) Let $|q|_2>\max\{|\theta-1|_2,|4m|_2\}$. If $\max\left\{\left|(\theta-1)^2\right|_2, \left|4m^2\right|_2\right\}\leq\left|4mq\right|_2$ then we have
$$\left|(\theta-1)^2D\right|_2\leq\left|(4mq)^2\right|_2.$$
Hence
$$|z_{1,2}(m)-1|_2=\frac{\left|(\theta-1)^2-2mq\pm(\theta-1)\sqrt{D}\right|_2}{\left|2m^2\right|_2}
=\frac{|2mq|_2}{\left|2m^2\right|_2}=\frac{|q|_2}{|m|_2}\geq1.$$
If $\max\left\{\left|(\theta-1)^2\right|_2, \left|4mq\right|_2\right\}\leq\left|4m^2\right|_2$ then we have
$$\left|(\theta-1)^2D\right|_2\leq\left|(4m^2)^2\right|_2$$
Consequently by non-Archimedean norm's property
$$|z_{1,2}(m)-1|_2\leq\frac{|2mq|_2}{\left|2m^2\right|_2}=\frac{|q|_2}{|m|_2}.$$
But from (\ref{z_1-1z_2-1}) we get
$$|(z_1(m)-1)(z_2(m)-1)|_2=\frac{\left|(\theta-1)^2-q^2\right|_2}{\left|m^2\right|_2}=\frac{\left|q^2\right|_2}{\left|m^2\right|_2}.$$
Thus we have
$$|z_{1,2}(m)-1|_2=\frac{|q|_2}{|m|_2}>\frac{1}{4}.$$
Let $\max\left\{\left|4mq\right|_2, \left|4m^2\right|_2\right\}<\left|(\theta-1)^2\right|_2$. If $|\theta-1|_2=|m|_2$ then from $|q|_2>|\theta-1|$ we get $\left|(\theta-1)^2\right|_2<|mq|_2$. Hence
$$|z_{1,2}(m)-1|_2=\frac{|2mq|_2}{\left|2m^2\right|_2}=\frac{|q|_2}{|m|_2}>\frac{1}{4}.$$
If $|\theta-1|_2>|m|_2$ then by Theorem \ref{tx2} there exists $\sqrt{D}$ if and only if $\left|(\theta-1)^2\right|_2<|mq|_2$. Consequently
$$|z_{1,2}(m)-1|_2=\frac{|2mq|}{\left|2m^2\right|_2}=\frac{|q|_2}{|m|_2}\geq\frac{1}{2}.$$
Thus we have shown that if $|q|_2>\max\{|\theta-1|_2,|4m|_2\}$ then the equation (\ref{kv}) has no solution in $\mathcal E_2\setminus\{1\}$.

4) Let $|4m|_2=|\theta-1|_2>|q|_2$. Then we have
$$D=4m^2\left(1+4\left(\frac{\theta-1}{4m}\right)^2-\frac{q}{m}\right)=4m^2(1+4+8\varepsilon)\qquad\mbox{where}\ |\varepsilon|_2\leq1.$$
By Theorem \ref{tx2} there does not exist $\sqrt{1+4+8\varepsilon}$ in $\mathbb Q_2$. Consequently, $\sqrt{D}$ does not exist in $\mathbb Q_2$.

5) Proof is a similar to the proof 5).

6) Let $|4m|_2=|\theta-1|_2=|q|_2$. In this case by Theorem \ref{tx2} there exists $\sqrt{D}$ in $\mathbb Q_2$ and $D=4m^2(1+2\varepsilon)^2$. For the solutions $z_{1,2}(m)$ we get
$$|z_{1,2}(m)-1|_2=\frac{\left|(\theta-1)^2-2mq\pm2m(\theta-1)(1+2\varepsilon)\right|_2}{\left|2m^2\right|_2}\leq\frac{|q|_2}{|m|_2}=\frac{1}{4}.$$
This means that $z_{1,2}\in\mathcal E_2\setminus\{1\}$.

7) $|m|_2=|\theta-1|_2=|q|_2>|4m|_2$ and $\theta\neq1\pm q$. In this case $\sqrt{D}$ exists and $\sqrt{D}=(\theta-1)(1+2\varepsilon)$. Consequently there exist solutions $z_{1,2}(m)$ in $\mathbb Q_2$.
It is easy to see that if $|a|_2=|b|_2$ than $|a\pm b|_2\leq|2a|_2$. Using this property to (\ref{z_1-1z_2-1}) we have
$$|(z_1(m)-1)(z_2(m)-1)|_2=\frac{\left|(\theta-1)^2-q^2\right|_2}{\left|m^2\right|_2}\leq\frac{\left|4m^2\right|_2}{\left|m^2\right|_2}=\frac{1}{4}.$$
Hence,
$$|z_1(m)-1|_2=1,\qquad |z_2(m)-1|_2\leq\frac{1}{4}\qquad\mbox{if}\ |\varepsilon|_2=1$$
and
$$|z_2(m)-1|_2=1,\qquad |z_1(m)-1|_2\leq\frac{1}{4}\qquad\mbox{if}\ |\varepsilon|_2<1.$$
This means that in this case the equation (\ref{kv}) has only one solution in $\mathcal E_2\setminus\{1\}$.

8) Let $|m|_2>|\theta-1|_2=|q|_2>|4m|_2$. If there exists $\sqrt{D}$ then it holds inequality $\left|\sqrt{D}\right|_2<\left|2(\theta-1)\right|_2$. From this
$$|z_{1,2}(m)-1|_2=\frac{\left|(\theta-1)^2-2mq\pm(\theta-1)\sqrt{D}\right|_2}{\left|2m^2\right|_2}\leq\frac{\left|(\theta-1)^2\right|_2}{\left|m^2\right|_2}=\frac{1}{4}.$$

9) Let $|\theta-1|_2=|q|_2>|m|_2$. In this case by Theorem \ref{tx2} there exists $\sqrt{D}$ and $\sqrt{D}=(\theta-1)(1+2\varepsilon)$ where $|\varepsilon|_2\leq1$.
If $|\varepsilon|_2=1$ we have
$$|z_2(m)-1|_2=\frac{\left|(\theta-1)^2\right|_2}{\left|m^2\right|_2}>1$$
and
$$|z_1(m)-1|_2=\frac{\left|(\theta-1)^2-q^2\right|_2}{\left|m^2\right|_2}\cdot\frac{\left|m^2\right|_2}{\left|(\theta-1)^2\right|_2}
\leq\frac{\left|4(\theta-1)^2\right|_2}{\left|(\theta-1)^2\right|_2}=\frac{1}{4}.$$
If $|\varepsilon|_2<1$ then we get
$$|z_1(m)-1|_2=\frac{\left|(\theta-1)^2\right|_2}{\left|m^2\right|_2}>1\qquad\mbox{and}\qquad |z_2(m)-1|_2\leq\frac{1}{4}.$$
Thus we have shown that in this case the equation (\ref{kv}) has only one solution in $\mathcal E_2\setminus\{1\}$.
\end{proof}
\begin{cor}\label{cor21}
Let $p=2$.
\begin{itemize}
\item[1)] Let $|q|_2>\frac{1}{4}$. If $|m|_2=1$ then the equation (\ref{kv}) has no solutions in $\mathcal E_p$.
\item[2)] Let $|q|_2=\frac{1}{4}$. If $|m|_2=1$ then the equation (\ref{kv}) has solution in $\mathcal E_p$ if and only if $|\theta-1|_2=\frac{1}{4}$. Furthermore the equation (\ref{kv}) has two solutions if $\theta\notin\{1-q,1+q\}$ and it has one solution if $\theta\in\{1-q,1+q\}$.
\item[3)] Let $|q|_2<\frac{1}{4}$. If $|m|_2=1$ then the equation has solution in $\mathcal E_p$ if and only if $|\theta-1|_2<\frac{1}{4}$. Furthermore the equation (\ref{kv}) has two solutions if $\theta\notin\{1-q,1+q\}$ and it has one solution if $\theta\in\{1-q,1+q\}$.
\end{itemize}
\end{cor}
By Corollary \ref{cor1} and by Propositions \ref{pro21}-\ref{pro23} we get the following
\begin{thm}\label{tigm2}
Let $p=2$.
\begin{itemize}
\item[1)] For a given $m\leq[q/2]$ there exist $2{q \choose m}$ of TIpGMs if at least one of the following conditions is satisfied\\
1a) $|4m|_2>\max\{|\theta-1|_2,|q|_2\}$ and $\theta\notin\{1-q,1+q\}$\\
1b) $|4m|_2=|\theta-1|_2=|q|_2$ and $\theta\notin\{1-q,1+q\}$\\
1c) $|m|_2>|\theta-1|_2=|q|_2>|4m|_2,\ q\neq2m,\ \theta\notin\{1-q,1+q\}$ and there exists $\sqrt{1-2a+b^2}$, where $a=\frac{q}{2m},\ b=\frac{\theta-1}{2m}$.\\
\item[2)] For a given $m\leq[q/2]$ there exist ${q \choose m}$ of TIpGMs if at least one of the following conditions is satisfied\\
2a) $|4m|_2>\max\{|\theta-1|_2,|q|_2\}$ and $\theta\in\{1-q,1+q\}$\\
2b) $|4m|_2=|\theta-1|_2=|q|_2$ and $\theta\in\{1-q,1+q$\\
2c) $|m|_2=|\theta-1|_2=|q|_2>|4m|_2$ and $\theta\notin\{1-q,1+q\}$\\
2d) $|m|_2>|\theta-1|_2=|q|_2>|4m|_2$ and $\theta\in\{1-q,1+q\}$\\
2e) $|\theta-1|_2=|q|_2>|m|_2$ and $\theta\notin\{1-q,1+q\}$\\
2f) $|m|_2>|\theta-1|_2=|q|_2>|4m|_2,\ q=2m,\ \theta\notin\{1-q,1+q\}$ and there exists $\sqrt{b^2-1}$, where $b=\frac{\theta-1}{q}$.\\

\item[3)] Otherwise there does not exist any TIpGM.
\end{itemize}
\end{thm}

\subsection{\bf Boundedness of translation-invariant $p$-adic Gibbs measures}

Now we shall study the problem of boundedness of translation-invariant $p$-adic Gibbs measures. Note that if $q\notin p\mathbb N$ then by Theorems \ref{tigm1},\ref{tigm2} there exists only one translation-invariant $p$-adic Gibbs measure $\mu_0$. In \cite{MRasos} it have been proven that $p$-adic Gibbs measure $\mu_0$ is bounded if and only if $q\notin p\mathbb N$. Assume that $m\in\{1,2,...,[q/2]\}$ and $z(m)\in\mathcal E_p\setminus{1}$ is a solution to the equation (\ref{rm}). We shall show that corresponding $p$-adic Gibbs measure $\mu_{h(m)}$ is not bounded. Since
$$\left|\mu_{h(m)}^{(n)}(\sigma)\right|_p=\frac{\left|\exp_p\left(H_n(\sigma)+\sum_{x\in W_n}h(m){\bf1}(\sigma(x)\leq m)\right)\right|_p}{\left|Z_{n,h(m)}\right|_p}=\frac{1}{\left|Z_{n,h(m)}\right|_p}$$
We shall show that
$$\left|Z_{n,h(m)}\right|_p\to 0,\qquad n\to\infty.$$
For the normalizing constant we have the following recurrence formula \cite{MRasos}
\begin{equation}\label{rec}
Z_{n+1,h}=A_{n,h}Z_{n,h},\qquad\mbox{where}\ \ A_{n,h}=\prod_{x\in W_n}a_h(x).
\end{equation}
For the solution $z(m)\in\mathcal E_p\setminus\{1\}$ to the equation (\ref{rm}) we have
$$a_{h(m)}(x)=(m(z(m)-1)+q+\theta-1)^2,\qquad\mbox{where}\ \ z(m)=\exp_p(h(m)).$$
Then by (\ref{rec}) we get
$$Z_{n+1,h(m)}=(m(z(m)-1)+q+\theta-1)^{2|V_n|}.$$
From this considering $|z(m)-1|_p<1,\ |\theta-1|_p<1$ and $q\in p\mathbb N$ we have
$$\left|Z_{n+1,h(m)}\right|_p<p^{-2|V_n|}.$$
Hence,
$$\left|Z_{n,h(m)}\right|_p\to 0,\qquad n\to\infty.$$
Thus we have proved the following
\begin{thm} Translation-invariant $p$-adic Gibbs measures for the Potts model on the Cayley tree of order two are bounded if and only if $q\notin p\mathbb N$.
\end{thm}

\subsection{\bf The number of TIpGMs} Denote by $\mathcal N_{TI}$ the number of all translation-invariant $p$-adic Gibbs measures for the $q$-state $p$-adic Potts model on the Cayley tree of order two. Note that $\mathcal N_{TI}$ depends on the parameter $\theta$ (since $\theta=\exp_p(J)$ it depends on $J$). Since the TIpGM $\mu_0$ exists independently on parameters,  the set of all TIpGMs is not empty.

{\bf 1)} Let $q\notin p\mathbb N$. In this case by Theorem \ref{tigm1} there exists a unique translation-invariant $p$-adic Gibbs measure $\mu_0$, i.e. $\mathcal N_{TI}=1$.

{\bf 2)} Let $q=p>2$ (If $q=p=2$ we get $2$-adic Ising model. It is known (see \cite{MRasos}) that for the Ising model there exists a unique $p$-adic Gibbs measure which is translation-invariant. So, $\mathcal N_{TI}=1$). Then for any integer number $m\in\{1,2,...,[q/2]\}$ it holds $|m|_p>|q|_p\geq|\theta-1|_p$. By Theorem \ref{tigm1}
for the integer number $m\leq[q/2]$ there are $2{q \choose m}$ of TIpGMs if $\theta\notin\{1-q,1+q\}$ and there are ${q \choose m}$ if $\theta\in\{1-q,1+q\}$.

Using
$${q \choose m}={q \choose q-m},\qquad \sum_{m=1}^q{q \choose m}=2^q-1$$
we get
$$\mathcal N_{TI}=1+2\sum_{m=1}^{[q/2]}{q \choose m}=2^q-1,\qquad\mbox{ if }\theta\notin\{1-q,1+q\}$$
and
$$\mathcal N_{TI}=1+\sum_{m=1}^{[q/2]}{q \choose m}=2^{q-1},\qquad\mbox{ if }\theta\in\{1-q,1+q\}.$$

{\bf 3)} Let $p>2$ and $q=pn,\ n\in\{2,p-1\}$. Then
$$|m|_p>|q|_p\geq|\theta-1|_p,\qquad\mbox{if}\ \ m\in\{1,2,...,[pn/2]\}\setminus\{p,2p,...,[n/2]p\}$$
and
$$|m|_p=|q|_p\geq|\theta-1|_p\qquad\mbox{if}\ \ m\in\{p,2p,...,[pn/2]\}.$$
By Theorem \ref{tigm1} similarly as proof of Proposition 2 in \cite{KRK}, one can show that
\[ \mathcal N_{TI}=\left\{\begin{array}{ll}
2^q-1-2\sum_{s=1}^{[n/2]}{pn\choose ps},&\mbox{ if }\ \theta\notin\{1-q,1+q\}\mbox{ and }q\mbox{ is } odd\\[3mm]
2^q-1+{q\choose [q/2]}-2\sum_{s=1}^{[n/2]}{pn\choose ps},&\mbox{ if }\ \theta\notin\{1-q,1+q\}\mbox{ and }q\mbox{ is } even\\[3mm]
2^{q-1}-\sum_{s=1}^{[n/2]}{pn\choose ps},&\mbox{ if }\ \theta\in\{1-q,1+q\}\mbox{ and }q\mbox{ is } odd\\[3mm]
2^{q-1}+{q\choose [q/2]}-\sum_{s=1}^{[n/2]}{pn\choose ps},&\mbox{ if }\ \theta\in\{1-q,1+q\}\mbox{ and }q\mbox{ is } even\\[3mm]
\end{array}\right.\]

{\bf 4)} Let $p>2$ and $q=p^sn$, where $s>1,\ n\in\{1,...,p-1\}$. If $n=1$ then there are at most $2^q-1$ of TIpGMs. Note that $\mathcal N_{TI}=2^q-1$ if and only if $0<\left|(\theta-1)^2-q^2\right|_p\leq\left|q^2\right|_p$.

If $1<n\leq p-1$ and $n$ is odd then there are at most $2^q-1-2\sum_{m=1}^{[n/2]}{p^sn\choose p^sm}$ of TIpGMs.

If $1<n\leq p-1$ and $n$ is even then there are at most $2^q-1+{q\choose [q/2]}-2\sum_{m=1}^{[n/2]}{p^sn\choose p^sm}$ of TIpGMs.

{\bf 5)} Let $p=2$ and $|q|_2>\frac{1}{4}$. Then by Theorem \ref{tigm2} there exists a unique TIpGMs. Thus, in this case $\mathcal N_{TI}=1$.

{\bf 6)} Let $p=2$ and $q=4$. Then there are at most 15 of TIpGMs. If $\sqrt{(\theta-5)(\theta+3)}$ exists in $\mathbb Q_2$ then there exist 15 of TIpGMs. By Theorem \ref{tx2} the number $\sqrt{(\theta-5)(\theta+3)}$ exists if and only if
$$
\theta\in\left\{x\in\mathbb Q_2:\ |x-29|_2\leq\frac{1}{128}\right\}\bigcup\left\{x\in\mathbb Q_2:\ |x-93|_2\leq\frac{1}{256}\right\}\bigcup
$$
$$\left\{x\in\mathbb Q_2:\ |x-165|_2\leq\frac{1}{256}\right\}\bigcup\bigcup_{s=1}^\infty\left\{x\in\mathbb Q_2:\ |x-5-2^s|_2\leq\frac{1}{2^{s+3}}\right\}.$$

\end{document}